\documentclass[12pt]{article}
\usepackage{amsmath, amsfonts, amssymb, amsthm, graphicx, geometry, arydshln, umoline, subfig, newtxtext, newtxmath, color, comment, soul, enumerate, bm,float,multirow}
\usepackage{appendix,natbib}
\usepackage{tikz}
\usetikzlibrary{matrix, positioning,arrows.meta,calc}
\usepackage[colorlinks=true,citecolor=blue]{hyperref}

\usepackage{booktabs}
\usepackage[table]{xcolor}
\def\BibTeX{{\rm B\kern-.05em{\sc i\kern-.025em b}\kern-.08em
    T\kern-.1667em\lower.7ex\hbox{E}\kern-.125emX}}
\usepackage{makecell}
\usepackage{bbm}

\theoremstyle{definition}
\newtheorem{theorem}{Theorem}

\newtheorem{lemma}{Lemma}

\newcommand{\argmin}{\mathop{\rm arg~min}\limits}

\begin{document}

\title{Note on an Axiomatization of the Baldwin Rule}
\author{
Leo Goto\thanks{Undergraduate School of Management, Department of Business Economics, Tokyo University of Science, 1-11-2, Fujimi, Chiyoda-ku, Tokyo, 102-0071, Japan. Email: leogotodeb@gmail.com}
\and
Satoshi Nakada\thanks{School of Management, Department of Business Economics, Tokyo University of Science, 1-11-2, Fujimi, Chiyoda-ku, Tokyo, 102-0071, Japan. Email: snakada@rs.tus.ac.jp}
}
\date{\today}
\maketitle

\begin{abstract}
\citet{GotoNakada2026} showed that the Baldwin rule can be characterized using \textit{Neutrality}, \textit{Strong Bottom Consistency}, \textit{Faithfulness}, \textit{Cancellation} and \textit{Bottom Independence}. 
While their proof relies on the technique of linear algebra and graph theory, in this note, we provide a simpler proof using purely combinatorial arguments based on permutations and amplified preference profiles, thereby providing a more transparent proof of the characterization.

\end{abstract}

\section{Introduction}
One of the fundamental desiderata of voting rules is the selection of the Condorcet winner: an alternative that strictly dominates every other alternative in pairwise comparisons.
A well-known implementation of this criterion is obtained by iteratively eliminating the alternatives with the lowest Borda score, as embodied in the Nanson rule \citep{Nanson1882} and the Baldwin rule \citep{Baldwin1926}.

The axiomatizations of the Baldwin rule have previously been obtained by \citet{FreemanBrillConitzer2014} and \citet{GotoNakada2026}.
In particular, the latter characterization relies on linear algebraic techniques in order to unify the proof with the characterization of the Nanson rule \citep{Nanson1882}.

In this note, we provide a purely permutational proof of the characterization of the Baldwin rule, thereby avoiding the linear-algebraic arguments used in \citet{GotoNakada2026}.
Our proof highlights more directly how each axiom contributes to the iterative elimination of alternatives with the lowest Borda score.

The remainder of the paper is organized as follows. 
Section \ref{sec: model} introduces our model.
Section \ref{sec: result} states our characterization theorem. 
Section \ref{sec: proof} provides the key combinatorial constructions and establishes the theorem. 

\section{Voting rule}\label{sec: model}
\subsection{Basic notation}
Let $\mathcal{N}=\{1,2,\ldots\}$ be a countably infinite set of potential voters.
Let $\mathcal U$ denote the collection of all finite and nonempty subsets of $N$.
An element $N\in\mathcal U$ is called an electorate.
Let $A=\{a_1,\ldots,a_m\}$ with $m\ge3$ be the universal set of alternatives.
Let $\mathcal B=\{B\subseteq A:B\neq\emptyset\}$ denote the collection of all nonempty feasible sets.
For every $B\in\mathcal B$, let $L(B)$ denote the set of strict linear orders on $B$.
For $\succ_i \, \in L(B)$ and $a\in B$, define the rank of $a$ under $\succ_i$ by
\[
r_B(\succ_i,a)
=
|\{b\in B:b\succ_i a\}|+1.
\]
For every nonempty $C\subseteq B$ and every $\succ_i\in L(B)$, let $\succ_i|_C\in L(C)$ denote the restriction of $\succ_i$ to $C$.
For every electorate $N\in\mathcal U$ and every feasible set $B\in\mathcal B$, define
$P_N(B)=\prod_{i\in N}L(B)$.
An element $\succ=(\succ_i)_{i\in N}\in P_N(B)$ is called a preference profile on $B$.
For every nonempty $C\subseteq B$, let $\succ\!\!|_C=(\succ_i|_C)_{i\in N}\in P_N(C)$
denote the restricted profile.

Let $N,N'\in\mathcal U$ be disjoint electorates.
For $\succ\in P_N(B)$ and $\succ'\in P_{N'}(B)$, their union is the profile $\succ\cup\succ'\in P_{N\cup N'}(B)$ defined by
\[
(\succ\cup\succ')_i=
\begin{cases}
\succ_i~\text{if}~i\in N,\\
\succ_i'~\text{if}~i\in N'.
\end{cases}
\]
For every positive integer $k$, let $k*\succ$ denote the $k$-fold replication of $\succ$, namely the union of $k$ pairwise disjoint copies of $\succ$.
Define $P(B)=\bigcup_{N\in\mathcal U}P_N(B)$ and $P=\bigcup_{B\in\mathcal B}P(B)$.
When no confusion can arise, we suppress the electorate superscript and
write simply $\succ$.

For every feasible set $B\in\mathcal B$, let $R(B)$ denote the set of weak orders on $B$.
Let $R=\bigcup_{B\in\mathcal B}R(B)$.
For $R\in R(B)$, let $P_R$ and $I_R$ denote its asymmetric and symmetric
parts, respectively.
A voting rule is a function $F:P\rightarrow R$ such that $F(\succ)\in R(B)$ for every profile $\succ\in P(B)$.
For a profile $\succ$, let $R_F(\succ)$ denote the weak order selected
by $F$.
Define the bottom tier of $F(\succ)$ by
\[
W_F(\succ)
=
\operatorname{bottom}(F(\succ))
=
\{a\in B:
bR_F(\succ)a
\text{ for every }b\in B
\}.
\]
If all alternatives are socially indifferent, we write $F(\succ)=(B)$.

\subsection{The Baldwin rule}
For every profile $\succ\in P(B)$ and every alternative $a\in B$, define the centered Borda score by
\[
\beta_a^B(\succ)
=
\sum_{b\in B\setminus\{a\}}
\left(
|\{i\in N:a\succ_i b\}|
-
|\{i\in N:b\succ_i a\}|
\right).
\]
Let $\beta^B(\succ)=(\beta_a^B(\succ))_{a\in B}$ denote the centered Borda-score vector.
When the feasible set is clear from the context, we simply write
$\beta_a(\succ)$ and $\beta(\succ)$.
Note that $\sum_{a\in B}\beta_a^B(\succ)=0$.
To relate the centered score to the standard positional Borda score,
define
\[
b_a^B(\succ)
=
\sum_{i\in N}
\left(
|B|-r_B(\succ_i,a)
\right).
\]
Then
\[
\beta_a^B(\succ)
=
2b_a^B(\succ)-|N|(|B|-1),
\]
so that the centered and positional Borda scores induce the same ranking.
The Borda rule, denoted by $F^{\rm Borda}$, ranks alternatives according to their centered Borda scores:
\[
aR_{F^{\rm Borda}(\succ)}b
\iff
\beta_a^B(\succ)\ge
\beta_b^B(\succ).
\]

The Baldwin rule is obtained by recursively eliminating the alternatives
with the lowest centered Borda score.
For every profile $\succ\,\in P(B)$, let
\[
B^0=B.
\]

For every round $t\ge0$, define
\[
E^t(\succ)
=
\arg\min_{a\in B^t}
\beta_a^{B^t}(\succ\!\!|_{B^t}).
\]

If $E^t(\succ)=B^t$, the procedure terminates.

Otherwise, let
\[
B^{t+1}
=
B^t\setminus E^t(\succ),
\]
and continue recursively.

For every alternative $a\in B$, define its elimination round
$\tau(a;\succ)$ by
\[
a\in E^{\tau(a;\succ)}(\succ).
\]

The Baldwin rule is then given by
\[
aR_{F^{\rm Baldwin}(\succ)}b
\iff
\tau(a;\succ)\ge
\tau(b;\succ).
\]

\section{An axiomatization of the Baldwin rule}\label{sec: result}

\cite{Young1974} characterized the Borda choice correspondence, and \cite{NitzanRubinstein1981} provided the corresponding characterization of the complete Borda ranking.
Following \cite{GotoNakada2026}, we consider the following recursive
counterparts of Young's axioms.

\bigskip

\noindent
\textbf{Neutrality.}
For every $B\in\mathcal B$, every $\succ\,\in P(B)$, and every $\sigma\in\Pi(A)$, $F(\sigma(\succ))=\sigma(F(\succ))$.

\bigskip

\noindent
\textbf{Strong Bottom Consistency.}
For every feasible set $B\in\mathcal B$, every pair of disjoint electorates $N,N'\in\mathcal U$, and every pair of profiles $\succ^N\in\mathcal P^N(B)$ and $\succ^{N'}\in\mathcal P^{N'}(B)$, if $W_F(\succ^N)\cap W_F(\succ^{N'})\neq\emptyset$,
then $W_F(\succ^N\cup\succ^{N'})=W_F(\succ^N)\cap W_F(\succ^{N'})$.

\bigskip

\noindent
\textbf{Faithfulness.}
For every one-voter profile $\succ\in P_{\{i\}}(B)$, $F(\succ)=\succ$.

\bigskip

\noindent
\textbf{Cancellation.}
For every $\succ\in P(B)$,
if $\beta^B(\succ)=\mathbf 0$, then $F(\succ)=(B)$.

\bigskip

\noindent
\textbf{Bottom Independence.}
For every $\succ\,\in P(B)$ such that $W_F(\succ)\neq B$, $F(\succ)|_{B\setminus W_F(\succ)}=F(\succ\!\!|_{B\setminus W_F(\succ)})$.

\textit{Neutrality} requires invariance with respect to the names of the
alternatives.
\textit{Strong Bottom Consistency} requires the current bottom tier to respond
consistently to electorate aggregation.
\textit{Faithfulness} requires the social ranking to coincide with the unique
voter's ranking.
\textit{Cancellation} requires complete social indifference whenever every
centered Borda score is zero.
\textit{Bottom Independence} requires the continuation ranking after the first
elimination to coincide with the ranking obtained by applying the rule
to the reduced profile.
The following theorem was established by \cite{GotoNakada2026}.

\begin{theorem}[\citealt{GotoNakada2026}]\label{main_baldwin}
A voting rule $F: \mathcal{P}\rightarrow \mathcal{R}$ satisfies \textit{Neutrality}, \textit{Strong Bottom Consistency}, \textit{Faithfulness}, \textit{Cancellation}, and \textit{Bottom Independence} if and only if it is the Baldwin rule, i.e., $F=F^{Baldwin}$.
\end{theorem}

\section{A proof}\label{sec: proof}
The proof proceeds in three steps.
First, in Subsection \ref{subsec: pref}, we introduce a specialized profile constructed by an amplification technique developed by \citet{HanssonSahlquist1976}, which will be useful to analyze the structural properties of the Borda scores associated with it.
Then, in Subsection \ref{subsec: mainproof}, we derive a sequence of implications of the axioms for these profiles. 
Finally, we combine these implications to show that the bottom-ranked alternatives must coincide with the alternatives attaining the minimum Borda score, from which the characterization follows by repeated applications of \textit{Bottom Independence}.

\subsection{Specific preference profiles}\label{subsec: pref}
The main technical ingredient is the amplified profile $\Sigma_B(\succ, C)$,  defined as follows.

Fix a feasible set $B\in\mathcal B$ with $q=|B|\ge3$.
For every $\succ\in P(B)$ and every nonempty $C\subseteq B$, define
\[
\Sigma_B(\succ,C)
=
\bigcup_{\substack{\sigma\in\Pi(B)\\
\sigma(x)=x\text{ for every }x\in C}}
\sigma(\succ).
\]
Thus, $\Sigma_B(\succ,C)$ is obtained by taking the union of all relabelings of $\succ$ that fix every alternative in $C$.
This construction averages over the alternatives in $B\setminus C$ while preserving the positions of the alternatives in $C$.

For every $z\in B$, its centered Borda score in the amplified profile is
\[
\beta_z^B(\Sigma_B(\succ,C))
=
\begin{cases}
(q-|C|)!\beta_z^B(\succ)
& \text{if }z\in C,\\[1mm]
-(q-|C|-1)!\displaystyle\sum_{x\in C}\beta_x^B(\succ)
& \text{if }z\in B\setminus C.
\end{cases}
\]
This follows from the symmetry of $\Sigma_B(\succ,C)$ with respect to all permutations of $B\setminus C$.

We next introduce two auxiliary profiles.
The first preserves the score difference between two selected alternatives, whereas the second equalizes their scores while lowering both relative to the remaining alternatives.

Fix $\succ\,\in P(B)$ and $x,z\in B$ such that
\[
\Delta_{xz}
=
\beta_z^B(\succ)-\beta_x^B(\succ)
>
0.
\]
Since $\Delta_{xz}=2\bigl(b_z^B(\succ)-b_x^B(\succ)\bigr)$, $\Delta_{xz}$ is even.

Define
\[
\succ_{\mathrm{diff},xz}
=
2\Delta_{xz}*
\succ_{\mathrm{diff},xz}^{\{i\}}
\]
for some $i\in N$, where
\[
r_B\bigl((\succ_{\mathrm{diff},xz}^{\{i\}})_i,x\bigr)
=
\begin{cases}
\frac{q+1}{2} & \text{if }q\text{ is odd},\\
\frac{q}{2}+1 & \text{if }q\text{ is even},
\end{cases}
\]
and
\[
r_B\bigl((\succ_{\mathrm{diff},xz}^{\{i\}})_i,z\bigr)
=
\begin{cases}
\frac{q-1}{2} & \text{if }q\text{ is odd},\\
\frac{q}{2} & \text{if }q\text{ is even}.
\end{cases}
\]

Define
\[
\succ_{\mathrm{level},xz}
=
l*\succ_{\mathrm{level},xz}^{\{i,j\}}
\]
for some $i,j\in N$, where
\[
r_B\bigl((\succ_{\mathrm{level},xz}^{\{i,j\}})_i,x\bigr)
=
r_B\bigl((\succ_{\mathrm{level},xz}^{\{i,j\}})_j,z\bigr)
=
\begin{cases}
\frac{q+3}{2} & \text{if }q\text{ is odd},\\
\frac{q}{2}+2 & \text{if }q\text{ is even},
\end{cases}
\]
\[
r_B\bigl((\succ_{\mathrm{level},xz}^{\{i,j\}})_i,z\bigr)
=
r_B\bigl((\succ_{\mathrm{level},xz}^{\{i,j\}})_j,x\bigr)
=
\begin{cases}
\frac{q+1}{2} & \text{if }q\text{ is odd},\\
\frac{q}{2}+1 & \text{if }q\text{ is even},
\end{cases}
\]
and
\[
r_B\bigl((\succ_{\mathrm{level},xz}^{\{i,j\}})_i,y\bigr)
=
r_B\bigl((\succ_{\mathrm{level},xz}^{\{i,j\}})_j,y\bigr)
\]
for every $y\in B\setminus\{x,z\}$.

Figures~\ref{fig:d_diff_odd}--\ref{fig:d_level_even} illustrate the positions of $x$ and $z$ in these profiles.
The alternatives in $B\setminus\{x,z\}$ such as $u,v,w$, and $y$ in the figures, may be placed arbitrarily.
\begin{figure}[H]
    \centering
    \begin{minipage}[b]{0.48\textwidth}
    \centering
    \begin{tikzpicture}

      \node (q) [
        matrix,
        matrix of nodes,
        nodes in empty cells,
        row sep=0.15cm,
        column sep=0.4cm,
        nodes={anchor=center},
        column 1/.style={nodes={align=right}},
        column 2/.style={nodes={align=center}},
        column 3/.style={nodes={align=center}}
      ] {
        & { \begin{tabular}{c}Borda \\ points\end{tabular} } & { $\succ^{\{i\}}_{\text{diff},xz}$ } \\
        1st place & $q-1$ & $u$ \\
        $\vdots$ & $\vdots$ & $\vdots$ \\
        $\frac{q-1}{2}$ nd place & $+2$ & $z$ \\
        $\frac{q+1}{2}$ nd place & $0$ & $x$ \\
        $\vdots$ & $\vdots$ & $\vdots$ \\
        $q$ th place & $1-q$ & $w$ \\
      };
      \coordinate (hy)  at ($(q-1-2.south)!0.5!(q-2-2.north) + (0, 0.1)$);
      \draw[thick] (q.west |- hy) -- (q.east |- hy);
      \coordinate (vx1) at ($(q-1-1.east)!0.5!(q-1-2.west) + (0.6, 0)$);
      \coordinate (vx2) at ($(q-1-2.east)!0.5!(q-1-3.west) + (-0.1,0)$);
      \coordinate (top_y) at ($(q-1-2.north) + (0, -0.2)$);
      \draw[thick] (vx1 |- top_y) -- (vx1 |- q.south);
      \draw[thick] (vx2 |- top_y) -- (vx2 |- q.south);
      \node at ($(q-4-3.east)!0.5!(q-5-3.east) + (0.7, 0)$) {$\dots$};
      \coordinate (arrow_start) at ($(q-1-3.north) + (0.2, 0)$);
      \coordinate (arrow_end)   at ($(q-1-3.north) + (1.6, -0.2)$);
      \draw[-{Latex[length=2mm, width=2mm]}, thick] (arrow_start) to[out=45, in=135] (arrow_end);
      \node[above, font=\footnotesize] at ($(arrow_start)!0.5!(arrow_end) + (0, 0.3)$) {Copy $2\Delta_{xz} - 1$ times};

    \end{tikzpicture}
    \caption{$\succ_{\text{diff},xz}$ where $q$ is odd}
    \label{fig:d_diff_odd}
\end{minipage}
    \hfill
    \begin{minipage}[b]{0.48\textwidth}
        \centering
        \begin{tikzpicture}

      \node (q) [
        matrix,
        matrix of nodes,
        nodes in empty cells,
        row sep=0.15cm,
        column sep=0.4cm,
        nodes={anchor=center},
        column 1/.style={nodes={align=right}},
        column 2/.style={nodes={align=center}},
        column 3/.style={nodes={align=center}}
      ] {
        & { \begin{tabular}{c}Borda \\ points\end{tabular} } & { $\succ^{\{i\}}_{\text{diff},xz}$ } \\
        1st place & $q-1$ & $u$ \\
        $\vdots$ & $\vdots$ & $\vdots$ \\
        $\frac{q}{2}$ nd place & $+1$ & $z$ \\
        $\frac{q}{2}+1$ st place & $-1$ & $x$ \\
        $\vdots$ & $\vdots$ & $\vdots$ \\
        $q$ th place & $1-q$ & $w$ \\
      };
      \coordinate (hy)  at ($(q-1-2.south)!0.5!(q-2-2.north) + (0, 0.1)$);
      \draw[thick] (q.west |- hy) -- (q.east |- hy);
      \coordinate (vx1) at ($(q-1-1.east)!0.5!(q-1-2.west) + (0.6, 0)$);
      \coordinate (vx2) at ($(q-1-2.east)!0.5!(q-1-3.west) + (-0.1,0)$);
      \coordinate (top_y) at ($(q-1-2.north) + (0, -0.2)$);
      \draw[thick] (vx1 |- top_y) -- (vx1 |- q.south);
      \draw[thick] (vx2 |- top_y) -- (vx2 |- q.south);
      \node at ($(q-4-3.east)!0.5!(q-5-3.east) + (0.7, 0)$) {$\dots$};
      \coordinate (arrow_start) at ($(q-1-3.north) + (0.2, 0)$);
      \coordinate (arrow_end)   at ($(q-1-3.north) + (1.6, -0.2)$);
      \draw[-{Latex[length=2mm, width=2mm]}, thick] (arrow_start) to[out=45, in=135] (arrow_end);
      \node[above, font=\footnotesize] at ($(arrow_start)!0.5!(arrow_end) + (0, 0.3)$) {Copy $2\Delta_{xz} - 1$ times};

    \end{tikzpicture}
    \caption{$\succ_{\text{diff},xz}$ where $q$ is even}
    \label{fig:d_diff_even}
    \end{minipage}
\end{figure}

\begin{figure}[htbp]
    \centering
    \begin{minipage}[b]{0.48\textwidth}
        \centering
        \begin{tikzpicture}

      \node (q) [
        matrix,
        matrix of nodes,
        nodes in empty cells,
        row sep=0.15cm,
        column sep=0.4cm,
        nodes={anchor=center},
        column 1/.style={nodes={align=right}},
        column 2/.style={nodes={align=center}},
        column 3/.style={nodes={align=center}},
        column 4/.style={nodes={align=center}}
      ] {
        & { \begin{tabular}{c}Borda \\ points\end{tabular} } & & \\
        1st place & $q-1$ & $u$ & $v$ \\
        $\vdots$ & $\vdots$ & $\vdots$ & $\vdots$ \\
        $\frac{q+1}{2}$ nd place & $0$ & $z$ & $x$ \\
        $\frac{q+3}{2}$ nd place & $-2$ & $x$ & $z$ \\
        $\vdots$ & $\vdots$ & $\vdots$ & $\vdots$ \\
        $q$ th place & $1-q$ & $w$ & $y$ \\
      };

      \coordinate (C34_center) at ($(q-2-3.north)!0.5!(q-2-4.north)$);
      \node (header) at (C34_center |- q-1-2.center) { $\succ^{\{i,j\}}_{\text{level},xz}$ };

      \coordinate (hy) at ($(q-1-2.south)!0.5!(q-2-2.north) + (0, 0.1)$);
      \coordinate (hline_R) at ($(q-2-4.east) + (0.2, 0)$);
      \draw[thick] (q.west |- hy) -- (hline_R |- hy);
      \coordinate (vx1) at ($(q-2-1.east)!0.5!(q-2-2.west) + (0.1, 0)$);
      \coordinate (vx2) at ($(q-2-2.east)!0.5!(q-2-3.west) + (-0.1, 0)$);
      \coordinate (top_y) at ($(q-1-2.north) + (0, -0.2)$);
      \draw[thick] (vx1 |- top_y) -- (vx1 |- q.south);
      \draw[thick] (vx2 |- top_y) -- (vx2 |- q.south);

      \coordinate (arrow_start) at ($(header.north) + (0.1, 0.1)$);
      \coordinate (arrow_end)   at ($(header.north) + (1.6, -0.1)$);
      \draw[-{Latex[length=2mm, width=2mm]}, thick] (arrow_start) to[out=35, in=145] (arrow_end);
      \node[above, font=\footnotesize] at ($(arrow_start)!0.5!(arrow_end) + (0, 0.2)$) {Copy $l - 1$ times};
      \coordinate (dots_y) at ($(q-4-4.east)!0.5!(q-5-4.east)$);
      \node at ($(dots_y) + (0.6, 0)$) {$\dots$};

    \end{tikzpicture}
        \caption{$\succ_{\text{level}, xz}$ where $q$ is odd}
        \label{fig:d_level_odd}
    \end{minipage}
    \hfill
    \begin{minipage}[b]{0.48\textwidth}
        \centering
        \begin{tikzpicture}

      \node (q) [
        matrix,
        matrix of nodes,
        nodes in empty cells,
        row sep=0.15cm,
        column sep=0.4cm,
        nodes={anchor=center},
        column 1/.style={nodes={align=right}},
        column 2/.style={nodes={align=center}},
        column 3/.style={nodes={align=center}},
        column 4/.style={nodes={align=center}}
      ] {
        & { \begin{tabular}{c}Borda \\ points\end{tabular} } & & \\
        1st place & $q-1$ & $u$ & $v$ \\
        $\vdots$ & $\vdots$ & $\vdots$ & $\vdots$ \\
        $\frac{q}{2} + 1$ st place & $-1$ & $z$ & $x$ \\
        $\frac{q}{2} + 2$ nd place & $-3$ & $x$ & $z$ \\
        $\vdots$ & $\vdots$ & $\vdots$ & $\vdots$ \\
        $q$ th place & $1-q$ & $w$ & $y$ \\
      };

      \coordinate (C34_center) at ($(q-2-3.north)!0.5!(q-2-4.north)$);
      \node (header) at (C34_center |- q-1-2.center) { $\succ^{\{i,j\}}_{\text{level},xz}$ };

      \coordinate (hy) at ($(q-1-2.south)!0.5!(q-2-2.north) + (0, 0.1)$);
      \coordinate (hline_R) at ($(q-2-4.east) + (0.2, 0)$);
      \draw[thick] (q.west |- hy) -- (hline_R |- hy);
      \coordinate (vx1) at ($(q-2-1.east)!0.5!(q-2-2.west) + (0.2, 0)$);
      \coordinate (vx2) at ($(q-2-2.east)!0.5!(q-2-3.west) + (-0.1, 0)$);
      \coordinate (top_y) at ($(q-1-2.north) + (0, -0.2)$);
      \draw[thick] (vx1 |- top_y) -- (vx1 |- q.south);
      \draw[thick] (vx2 |- top_y) -- (vx2 |- q.south);

      \coordinate (arrow_start) at ($(header.north) + (0.1, 0.1)$);
      \coordinate (arrow_end)   at ($(header.north) + (1.6, -0.1)$);
      \draw[-{Latex[length=2mm, width=2mm]}, thick] (arrow_start) to[out=35, in=145] (arrow_end);
      \node[above, font=\footnotesize] at ($(arrow_start)!0.5!(arrow_end) + (0, 0.2)$) {Copy $l - 1$ times};
      \coordinate (dots_y) at ($(q-4-4.east)!0.5!(q-5-4.east)$);
      \node at ($(dots_y) + (0.6, 0)$) {$\dots$};

\end{tikzpicture}
        \caption{$\succ_{\text{level}, xz}$ where $q$ is even}
        \label{fig:d_level_even}
    \end{minipage}
\end{figure}

By construction,
\[
\beta_z^B(\succ_{\mathrm{diff},xz})
-
\beta_x^B(\succ_{\mathrm{diff},xz})
=
4\Delta_{xz},
\]
and
\[
\beta_x^B(\succ_{\mathrm{level},xz})
=
\beta_z^B(\succ_{\mathrm{level},xz})
=
\begin{cases}
-2l & \text{if }q\text{ is odd},\\
-4l & \text{if }q\text{ is even}.
\end{cases}
\]

\begin{lemma}\label{lemma_difference}
Fix $B\in\mathcal B$ with $|B|\ge3$, $\succ\in P(B)$, and $x,z\in B$ such that $0>\beta_z^B(\succ)>\beta_x^B(\succ)$.
Then there exists $l\in\mathbb N$ such that
\[
\beta_y^B\bigl(
\succ_{\mathrm{diff},xz}
\cup
(l*\succ_{\mathrm{level},xz})
\bigr)
=
\beta_y^B(4*\succ)
\]
for every $y\in\{x,z\}$, and
\begin{align*}
&\beta_z^B\bigl(
\succ_{\mathrm{diff},xz}
\cup
(l*\succ_{\mathrm{level},xz})
\bigr)
-
\beta_x^B\bigl(
\succ_{\mathrm{diff},xz}
\cup
(l*\succ_{\mathrm{level},xz})
\bigr)\\
&\qquad
=
4\Delta_{xz}
=
\beta_z^B(4*\succ)-\beta_x^B(4*\succ).
\end{align*}
\end{lemma}

\begin{proof}
By construction,
\[
\beta_z^B(\succ_{\mathrm{diff},xz})
=
\begin{cases}
4\Delta_{xz} & \text{if }q\text{ is odd},\\
2\Delta_{xz} & \text{if }q\text{ is even}.
\end{cases}
\]
If $q$ is odd, set $l=2\Delta_{xz}-2\beta_z^B(\succ)>0$, and if $q$ is even, set $l=\frac{1}{2}
\bigl(\Delta_{xz}-2\beta_z^B(\succ)\bigr)>0$.
In either case, $\beta_z^B\bigl( \succ_{\mathrm{diff},xz} \cup (l*\succ_{\mathrm{level},xz}) \bigr) = \beta_z^B(4*\succ)$.
Since $\beta_x^B(\succ_{\mathrm{level},xz})=\beta_z^B(\succ_{\mathrm{level},xz})$, the remaining assertions follow.
\end{proof}

\subsection{Proof of Theorem \ref{main_baldwin}}\label{subsec: mainproof}

The proof of Theorem \ref{main_baldwin} proceeds through the following lemmas.
Lemma \ref{lemma_bottomsum} shows an implication of \textit{Neutrality} and \textit{Strong Bottom Consistency}.
It shows that, once an alternative belongs to a proper bottom tier, the amplification procedure isolates it as the unique bottom alternative.

\begin{lemma}\label{lemma_bottomsum}
Suppose that $F$ satisfies \textit{Neutrality} and \textit{Strong Bottom Consistency}.
For every $B\in\mathcal B$, every $\succ\,\in P(B)$, and every $x\in W_F(\succ)\subsetneq B$,
we have $W_F(\Sigma_B(\succ,\{x\}))=\{x\}$.
\end{lemma}

\begin{proof}
Fix $B\in\mathcal B$, $\succ\,\in P(B)$, and
$x\in W_F(\succ)\subsetneq B$.
\textit{Neutrality} implies $x\in W_F(\sigma(\succ))$ for every $\sigma\in\Pi(B)$ satisfying $\sigma(x)=x$.
Repeated applications of \textit{Strong Bottom Consistency} therefore give $x\in
W_F(\Sigma_B(\succ,\{x\}))$.

Fix $y\in B\setminus\{x\}$.
Since $W_F(\succ)\subsetneq B$, there exists $w\in B\setminus W_F(\succ)$.
Choose a permutation $\sigma\in\Pi(B)$ such that $\sigma(x)=x$ and $\sigma(w)=y$.
\textit{Neutrality} gives $y\notin W_F(\sigma(\succ))$.
\textit{Strong Bottom Consistency} then implies $y\notin W_F(\Sigma_B(\succ,\{x\}))$.
Thus, $W_F(\Sigma_B(\succ,\{x\}))=\{x\}$.
\end{proof}

The next lemma shows that the bottom tier depends only on the centered Borda-score vector.

\begin{lemma}\label{lemma_bottomborda1}
Suppose that $F$ satisfies \textit{Strong Bottom Consistency} and \textit{Cancellation}.
For every $B\in\mathcal B$ and every $\succ,\succ'\in P(B)$, if $\beta^B(\succ)=\beta^B(\succ')$,
then $W_F(\succ)=W_F(\succ')$.
\end{lemma}

\begin{proof}
Fix $B\in\mathcal B$ and $\succ,\succ'\in P(B)$ such that $\beta^B(\succ)=\beta^B(\succ')$.
Without loss of generality, we may assume that the two profiles have disjoint electorates.
Let $\succ^{-1}$ be a profile on a third, disjoint electorate obtained by reversing every preference order in $\succ$:
\[
x\succ_i^{-1}y
\iff
y\succ_i x.
\]
By construction, $\beta^B(\succ\cup\succ^{-1})=\beta^B(\succ'\cup\succ^{-1})=\mathbf0_B$.
\textit{Cancellation} gives $W_F(\succ\cup\succ^{-1})=W_F(\succ'\cup\succ^{-1})=B$.
\textit{Strong Bottom Consistency} therefore implies $W_F(\succ)=W_F(\succ\cup\succ^{-1}\cup\succ')=W_F(\succ')$.
\end{proof}

\begin{lemma}\label{lemma_bottomborda2}
Suppose that $F$ satisfies \textit{Neutrality}, \textit{Strong Bottom Consistency}, and \textit{Cancellation}.
For every $B\in\mathcal B$, every $\succ\,\in P(B)$, and every $x\in B$, if $\beta_x^B(\succ)>0$, then $x\notin W_F(\succ)$.
\end{lemma}

\begin{proof}
Suppose, to the contrary, that there exist $B\in\mathcal B$, $\succ\,\in P(B)$, and $x\in B$ such that $\beta_x^B(\succ)>0$ and $x\in W_F(\succ)$.
Since $\beta_x^B(\succ)>0$, we have $W_F(\succ)\subsetneq B$.
Lemma~\ref{lemma_bottomsum} gives $W_F(\Sigma_B(\succ,\{x\}))=\{x\}$.

Let $(\Sigma_B(\succ,\{x\}))^{-1}$ be the inverse of$\Sigma_B(\succ,\{x\})$.
Then
\[
\beta^B\bigl(
\Sigma_B(\succ,\{x\})
\cup
(\Sigma_B(\succ,\{x\}))^{-1}
\bigr)
=
\mathbf0_B.
\]
Hence, by \textit{Cancellation},
\[
W_F\bigl(
\Sigma_B(\succ,\{x\})
\cup
(\Sigma_B(\succ,\{x\}))^{-1}
\bigr)
=
B.
\]

We now construct a profile with the same Borda-score vector as $\Sigma_B(\succ,\{x\})$.
Let $\operatorname{rep}^x$ be a replication of a two-voter profile in which both voters place $x$ first and rank the alternatives in $B\setminus\{x\}$ in reverse orders.
Choose the number of replications so that $\beta^B(\operatorname{rep}^x)=
\beta^B(\Sigma_B(\succ,\{x\}))$.
By Lemma~\ref{lemma_bottomborda1}, $W_F(\operatorname{rep}^x)=W_F(\Sigma_B(\succ,\{x\}))=\{x\}$.

Let $\operatorname{rep}^{-x}$ be the inverse profile of
$\operatorname{rep}^x$.
Every voter in $\operatorname{rep}^{-x}$ places $x$ last.
\textit{Faithfulness} and repeated applications of \textit{Strong Bottom Consistency} give $W_F(\operatorname{rep}^{-x})=\{x\}$.
Moreover, $\beta^B(\operatorname{rep}^{-x})=\beta^B((\Sigma_B(\succ,\{x\}))^{-1})$.
Lemma~\ref{lemma_bottomborda1} therefore implies
\[
W_F((\Sigma_B(\succ,\{x\}))^{-1})=\{x\}.
\]
\textit{Strong Bottom Consistency} now gives
\[
W_F\bigl(
\Sigma_B(\succ,\{x\})
\cup
(\Sigma_B(\succ,\{x\}))^{-1}
\bigr)
=
\{x\},
\]
contradicting \textit{Cancellation}.
\end{proof}

\begin{lemma}\label{lemma_bottomborda3}
Suppose that $F$ satisfies \textit{Neutrality}, \textit{Strong Bottom Consistency}, and \textit{Cancellation}.
For every $B\in\mathcal B$, every $\succ\in P(B)$, and every $x\in B$, if
$\beta_x^B(\succ)=0$ and $x\in W_F(\succ)$, then $\beta_y^B(\succ)=0$ for every $y\in B$.
\end{lemma}

\begin{proof}
Suppose, to the contrary, that there exist $B\in\mathcal B$, $\succ\in P(B)$, and $x,y\in B$ such that $\beta_x^B(\succ)=0$,
$x\in W_F(\succ)$, and $\beta_y^B(\succ)\neq0$.
Since $\sum_{z\in B}\beta_z^B(\succ)=0$,
we may choose $y$ such that $\beta_y^B(\succ)>0$.
Lemma~\ref{lemma_bottomborda2} gives $y\notin W_F(\succ)$, and hence $W_F(\succ)\subsetneq B$.
Lemma~\ref{lemma_bottomsum} therefore implies
\[
W_F(\Sigma_B(\succ,\{x\}))=\{x\}.
\]

On the other hand, the amplification formula and $\beta_x^B(\succ)=0$ imply $\beta_z^B(\Sigma_B(\succ,\{x\}))=0$ for every $z\in B$.
\textit{Cancellation} gives
\[
W_F(\Sigma_B(\succ,\{x\}))=B,
\]
a contradiction.
\end{proof}

The preceding lemmas show that the bottom tier depends only on the centered Borda-score vector and cannot contain an alternative with a positive score.
We now identify it exactly.

\begin{lemma}\label{lemma_bottom}
Suppose that $F$ satisfies \textit{Neutrality}, \textit{Strong Bottom Consistency}, \textit{Faithfulness}, and \textit{Cancellation}.
Then, for every $B\in\mathcal B$ and every $\succ\,\in P(B)$, $W_F(\succ)=\argmin_{x\in B}\beta_x^B(\succ)$.
\end{lemma}

\begin{proof}
Fix $B\in\mathcal B$ and $\succ\,\in P(B)$.
If $|B|=1$, the conclusion is immediate.
Suppose that $B=\{x,z\}$.
If $\beta_x^B(\succ)=\beta_z^B(\succ)=0$,
the conclusion follows from \textit{Cancellation}.
Otherwise, cancel pairs of voters having opposite rankings.
The remaining profile consists only of copies of the same one-voter order.
\textit{Faithfulness} and repeated applications of \textit{Strong Bottom Consistency} imply that its bottom tier is the alternative with the lower centered Borda score.
Adding back the cancelling pairs does not change the bottom tier by \textit{Strong Bottom Consistency} and \textit{Cancellation}.
Thus, the conclusion also holds when $|B|=2$.

Suppose henceforth that $|B|\ge3$.
We prove both inclusions.

\medskip
\noindent
\textit{Step 1:} $W_F(\succ)\subseteq \argmin_{x\in B}\beta_x^B(\succ)$.

Suppose, to the contrary, that there exists $z\in W_F(\succ)$ such that $\beta_z^B(\succ)>\min_{y\in B}\beta_y^B(\succ)$.
Fix $x\in\argmin_{y\in B}\beta_y^B(\succ)$.
Lemma~\ref{lemma_bottomborda2} implies that
$\beta_z^B(\succ)\le0$.
If $\beta_z^B(\succ)=0$, Lemma~\ref{lemma_bottomborda3} would imply
$\beta^B(\succ)=\mathbf0_B$, contrary to $\beta_z^B(\succ)>\beta_x^B(\succ)$.
Hence, $0>\beta_z^B(\succ)>\beta_x^B(\succ)$.
By Lemma~\ref{lemma_difference}, there exists $l\in\mathbb N$ such that, for
$\succ'=\succ_{\mathrm{diff},xz}\cup(l*\!\succ_{\mathrm{level},xz})$, we have
\[
\beta^B(\Sigma_B(\succ',\{x,z\}))
=
4\beta^B(\Sigma_B(\succ,\{x,z\})).
\]
Lemma~\ref{lemma_bottomborda1} therefore gives $W_F(\Sigma_B(\succ',\{x,z\}))=W_F(\Sigma_B(\succ,\{x,z\}))$.
By construction,
\[
W_F(\Sigma_B(\succ_{\mathrm{diff},xz},\{x,z\}))
=
\{x\}
\]
and
\[
W_F(\Sigma_B(\succ_{\mathrm{level},xz},\{x,z\}))
=
\{x,z\}.
\]
Since $\Sigma_B(\succ',\{x,z\})=\Sigma_B(\succ_{\mathrm{diff},xz},\{x,z\})\cup\Sigma_B(l*\!\succ_{\mathrm{level},xz},\{x,z\})$, \textit{Strong Bottom Consistency} yields
\[
W_F(\Sigma_B(\succ,\{x,z\}))
=
\{x\}.
\]
However, \textit{Neutrality} implies $z\in W_F(\sigma(\succ))$ for every $\sigma\in\Pi(B)$ that fixes $z$.
In particular, this holds for every permutation fixing both $x$ and $z$.
\textit{Strong Bottom Consistency} therefore gives
\[
z\in W_F(\Sigma_B(\succ,\{x,z\})),
\]
a contradiction.

\medskip
\noindent
\textit{Step 2:} $\argmin_{x\in B}\beta_x^B(\succ)\subseteq W_F(\succ)$.

Fix $z\in\argmin_{y\in B}\beta_y^B(\succ)$.
Since $W_F(\succ)\neq\emptyset$, Step 1 implies that every $x\in W_F(\succ)$ satisfies $\beta_x^B(\succ)=\beta_z^B(\succ)$.
If $\min_{y\in B}\beta_y^B(\succ)=0$, then
$\beta^B(\succ)=\mathbf0_B$, and \textit{Cancellation} gives $W_F(\succ)=B$.
Suppose instead that $\beta_x^B(\succ)=\beta_z^B(\succ)=\min_{y\in B}\beta_y^B(\succ)<0$ for some $x\in W_F(\succ)$.
By the amplification formula, $\beta_x^B(\Sigma_B(\succ,\{x,z\}))=\beta_z^B(\Sigma_B(\succ,\{x,z\}))$ and every alternative in $B\setminus\{x,z\}$ has a strictly higher score.
Step 1 therefore gives
\[
W_F(\Sigma_B(\succ,\{x,z\}))
\subseteq
\{x,z\}.
\]
Moreover, \textit{Neutrality} and \textit{Strong Bottom Consistency} imply
$x\in W_F(\Sigma_B(\succ,\{x,z\}))$.
If $z\notin W_F(\succ)$, there exists a permutation fixing $x$ and $z$ under which $z$ is not in the bottom tier.
\textit{Strong Bottom Consistency} would then imply $z\notin W_F(\Sigma_B(\succ,\{x,z\}))$, so that
\[
W_F(\Sigma_B(\succ,\{x,z\}))=\{x\}.
\]
Interchanging $x$ and $z$ in the same construction gives the opposite conclusion.
Therefore,
\[
z\in W_F(\succ).
\]
\end{proof}

\begin{proof}[Proof of Theorem~\ref{main_baldwin}]
The ``if'' direction is immediate.
For the converse, suppose that $F$ satisfies all the axioms.

Fix $B\in\mathcal B$ and $\succ\,\in P(B)$.
Let $B^0=B$.
By Lemma~\ref{lemma_bottom}, $W_F(\succ)=\argmin_{a\in B^0}\beta_a^{B^0}(\succ|_{B^0})=E^0(\succ)$.
If $W_F(\succ)=B^0$, the procedure terminates and both rules rank all alternatives in $B^0$ indifferently.
Otherwise, $B^1=B^0\setminus W_F(\succ)$,
and \textit{Bottom Independence} gives
\[
F(\succ)|_{B^1}
=
F(\succ\!\!|_{B^1}).
\]

Proceed inductively.
Suppose that $B^t$ has been obtained.
Lemma~\ref{lemma_bottom}, applied to the profile $\succ\!\!|_{B^t}\in P(B^t)$, gives $W_F(\succ\!\!|_{B^t})=\argmin_{a\in B^t}\beta_a^{B^t}(\succ\!\!|_{B^t})=E^t(\succ)$.
If $E^t(\succ)=B^t$, the the recursion terminates.
Otherwise, $B^{t+1}=B^t\setminus E^t(\succ)$, and \textit{Bottom Independence} yields
\[
F(\succ)|_{B^{t+1}}
=
F(\succ\!\!|_{B^{t+1}}).
\]
Thus, at every round, the bottom tier selected by $F$ coincides with the elimination set of the Baldwin rule.
Therefore,
\[
F(\succ)=F^{\rm Baldwin}(\succ).
\]
Since $B$ and $\succ$ were arbitrary, $F=F^{\rm Baldwin}$.
\end{proof}

\bibliographystyle{ecta}
\bibliography{ref}

\end{document}